\documentclass[a4paper, 10pt, conference]{ieeeconf}

\IEEEoverridecommandlockouts                              
\usepackage{amsmath,amssymb}
\usepackage{dsfont}
\usepackage{color}
\usepackage{multirow}
\usepackage{algorithm}
\usepackage{graphicx}      
\newcommand{\ver}{}

\newcommand{\comenta}[1]{}

\title{\LARGE \bf
On the linear quadratic problem for systems with time reversed Markov jump parameters 
and the duality with filtering of Markov jump linear systems}

\author{Daniel  Gutierrez and  Eduardo F. Costa 
\thanks{This work was supported by FAPESP, CNPQ and Capes. 
		Authors are with the 
        Instituto de Ci\^encias Matem\'aticas e de Computa\c{c}\~ao,
        Universidade de S\~ao Paulo, Brasil
        {\tt\small dgutip@icmc.usp.br,\; efcosta@icmc.usp.br}}%
}
\newtheorem{defi}{Definition}[section]
\newtheorem{teo}{Theorem}[section]

\newtheorem{lema}{Lemma}[section]

\newtheorem{rema}{Remark}

\begin{document}

\maketitle
\thispagestyle{empty}
\pagestyle{empty}

\begin{abstract}
We study a class of systems whose parameters are driven by a Markov chain in reverse time. 
A recursive characterization for the second moment matrix, a spectral radius test 
for mean square stability and the formulas for optimal control are given.  
Our results are determining for the question: is it possible to extend the classical duality 
between filtering and control of linear systems (whose matrices 
are transposed in the dual problem) by simply adding the jump variable
of a Markov jump linear system. The answer is positive provided the jump process 
is reversed in time.
\end{abstract}

\section{Introduction}
In this note we study a class of systems whose parameters are driven by a 
time reversed Markov chain. Given a time horizon $\ell$ and 
a standard Markov chain $\{\eta(t), t=0,1,\ldots\}$ taking values in the 
set $\{1,2,\ldots,N\}$, we consider the process
\begin{equation}\label{eq-def-theta}
\theta(t)=\eta(\ell-t), \quad 0\leq t\leq \ell,
\end{equation}
and the system
\begin{equation}\label{eq-def-sys}
\Phi:\begin{cases}
x(t+1)&=\,\, A_{\theta(t)}x(t)+ B_{\theta(t)}u(t),\\
\quad x(0)&=\,\, x_0, \; 0\leq t\leq \ell-1,
\end{cases}
\end{equation}
where, as usual, $x$ represents the state variable of the system and 
$u$ is the control variable.
These systems may be encountered in real world problems, specially  
when a Markov chain interacts with the system parameters via 
a \emph{first in last out} queue. 
An example consists of drilling sedimentary rocks whose 
layers can be modelled by a Markov chain from bottom to top 
as a consequence of their formation process.  
The first drilled layer is the last formed one. 
Another example is a DC-motor whose brush is grind 
by a machine subject to failures, leaving a series of imprecisions on the brush 
width that can be described by a Markov chain,
so that the last failure will be the first to affect  
the motor collector depending on how the brush is installed. 

One of the most remarkable features of system $\Phi$ is that 
it provides a dual for optimal filtering of standard \emph{Markov jump linear systems} (MJLS). 
In fact, if we consider a quadratic cost functional for system $\Phi$ with 
linear state feedback, leading to an optimal control problem \cite{Riccati-Basar} that we call  
\emph{time reversed Markov jump linear quadratic} problem (TRM-JLQP), then 
we show that the solution is identical to the gains of the 
\emph{linear minimum mean square estimator} (LMMSE) formulated 
in \cite{CostaFragosoMarques05,Tuesta03}, with time-reversed gains and transposed matrices.
In perspective with existing duality relations, the one obtained here 
is a direct generalization of the well known 
relation between control and filtering of linear time varying systems 
as presented for instance in \cite[Table 6.1]{Davis84}, or also in 
\cite{Panos,Kalman-Bucy1961,Song2013} in different contexts.
\comenta{For standard linear system \cite{Panos}, this is $N=1$.  
The duality is described by a reversal of solution time and a 
transposition of system matrices. 
In the  \cite[Equation 16] {Kalman-Bucy1961} we show this notion of duality. 
In the context of delay system, we can seen duality in
\cite{Song2013}.} 
As for MJLS, the duality between control and filtering have been considered  
e.g. in \cite{Abou95,Bittanti91,Costa95,CostaFragosoMarques05,Drag08,Dragan2009,Abou2003}, 
while purely in the context of standard MJLS,
thus leading to more complex relations involving certain generalized 
coupled Riccati difference  equations. 
Here, the duality follows naturally from the simple reversion of the 
Markov chain given in \eqref{eq-def-theta}, 
with no extra assumptions nor complex constructions. 

Another interesting feature of $\Phi$ is that the variable 
$\mathbb{E}\{x(t)x(t)'\cdot \mathds{1}_{\{\theta(t)=i\}}\}$, 
which is commonly used in the literature of MJLS, \cite{Vargas-Motor,CostaFragosoMarques05,doVal99JE,Eduardo-Vargas}, 
evolves along time $t$ according to a \emph{time-varying} linear operator, 
as shown in Remark \ref{Rem_evolution_first_moment}, 
in a marked dissimilarity with standard MJLS. 
This motivated us to employ $\mathfrak{X}(k)$, the conditioned second moment of $x(k)$, 
leading to time-homogeneous operators. 
\comenta{It is worth a mention that $\Phi$ is equivalent to an MJLS if, and only if, 
the Markov chain is revertible, 
see Remark \ref{Rem_evolution_first_moment} for details.} 

The contents of this note are as follows. We present basic notation in 
Section \ref{sec-notation}. In Section \ref{sec-sys} we give the recursive 
equation describing $\mathfrak{X}$, 
which leads to a stability condition involving the spectral radius of a 
time-homogeneous linear operator. 
In Section \ref{sec-problem}, we formulate and solve the TRM-JLQ problem, 
following a proof method where we decompose $\mathfrak{X}$ 
into two components as to handle $\theta$ that are 
visited with zero probability.
The duality with the LMMSE then follows in a straightforward manner, 
as presented in Section \ref{sec-duality}. Concluding remarks 
are given in Section \ref{sec-conclusions}.

\section{Notation and the system setup}\label{sec-notation}
Let $\Re^{n}$ be the $n$-dimensional euclidean space and 
$\Re^{m,n}$ be the space formed by real matrices of dimension $m$ by $n$. 
We write $\mathcal{C}^{m,n}$ to represent the Hilbert space composed of $N$ real matrices, that is $Y=(Y_1,\ldots,Y_N)\in\mathcal{C}^{m,n}$, where 
$Y_i\in \Re^{m,n}$, $i=1,\ldots,N$. 
The space $\mathcal{C}^{n,m}$ equipped with the inner product
$\langle Y,Z\rangle=\sum_{i=1}^{N}\text{Tr}(Y_{i}'Z_{i}),$
where $\text{Tr}(\cdot)$ is the trace operator and the superscript $'$  denotes the transpose, is a Hilbert space. 
The inner product induces the norm $\| Y \|= \langle Y,Y\rangle^{1/2}$.
If $n=m$, we write simply  $\mathcal{C}^{n}$. 
The mathematical operations involving elements of $\mathcal{C}^{n,m}$,
are used in element-wise fashion, e.g. for $Y$ and $Z$ in $\mathcal{C}^{n}$ we have
$YZ=\, (Y_1Z_1,\ldots,Y_NZ_N)$,
where $Y_iZ_i$ is the usual matrix multiplication.
Similarly, for a set of scalars  
$\alpha=(\alpha_1,\ldots,\alpha_N)\in \mathcal{C}^{1}$ we write
$\alpha Y=\,\,\, (\alpha_1 Y_1,\ldots, \alpha_NY_N)$. 

Regarding the system setup, it is assumed throughout the paper that
$x_0\in\Re^{n}$ is a random variable with zero mean satisfying $\mathbb E\{x_0x_0'\}=\Delta$.
We have $x(t)\in\Re^{n}$ and $u(t)\in\Re^{m}$.\comenta{$y(t)\in\Re^{s}$, $t=0,\ldots, \ell$.}
The system matrices belong to given sets  
$A,E \in\mathcal{C}^{n}$, $B\in\mathcal{C}^{n,m}$, 
$C\in\mathcal{C}^{s,n}$ and $D\in\mathcal{C}^{s,m}$ with
$C'_iD_i=0$ and $D'_iD_i>0$ for each $i=1,\ldots,N$.  
We write $\pi_i(t)=\text{Pr}(\theta(t)=i)$, where $\text{Pr}(\cdot)$  is the probability measure; $\pi(t)$ is considered as an element of $\mathcal{C}^1$, 
that is, $\pi(t)=(\pi_1(t),\ldots,\pi_N(t))$. 
$\pi_i$ stands for the limiting distribution of the Markov chain $\eta$ 
when it exists, in such a manner that $\pi_i=\lim_{\ell\rightarrow \infty}\pi_i(0)$. 
Also, we denote by $\mathcal{P}=[p_{ij}]$, $i,j=1,\ldots,N$
the transition probability matrix of the Markov chain $\eta$, so that for any $t=1,\ldots, \ell$, 
$$\begin{aligned}
&\text{Pr}(\theta(t-1)=j \,|\,\theta(t)=i)=\\
&\;= \text{Pr}(\eta(\ell-t+1)=j \,|\,\eta(\ell-t)=i)
= p_{ij}.
\end{aligned}$$
No additional assumption is made on the Markov chain, yielding a 
rather general setup that includes periodic chains, 
important for the duality relation given in Remark \ref{rema-extending-to-time-varying}.  
We shall deal with linear operators
$\mathcal{U}_{Z},\mathcal{V}_{Z},\mathcal{D}\,:\,\mathcal{C}^{n} \to \mathcal{C}^{n}$. 
We write the i-th element of $\mathcal{U}_{Z}(Y)$ by $\mathcal{U}_{Z,i}(Y)$
and similarly for the other operators. For each $i=1,\ldots,N$, we define:
\begin{equation}\label{eq-operators}
\begin{aligned}
\mathcal{U}_{Z,i}(Y)&=\,\,\,\displaystyle\sum_{j=1}^{N}p_{ij}Z_jY_jZ'_j\\
\mathcal{V}_{Z,i}(Y)&=\,\,\,Z'_i\mathcal{D}_i(Y)Z_i,\\
\,\,\,\mathcal{D}_i(Y)&=\,\,\,\displaystyle\sum_{j=1}^{N}p_{ji}Y_j.
\end{aligned}
\end{equation}
\section{Properties of system $\Phi$}\label{sec-sys}
Let $\mathbb{E}\{\cdot·\}$ be the expected value of a random
variable. We consider the \emph{conditioned second moment} of $x(t)$ defined by 
\begin{equation}\label{eq-def-cond-second-moment}
\mathfrak{X}_{i}(t)=\mathbb{E}\{x(t)x(t)'\,|\,\theta(t)=i\},\quad t=0,1,\ldots.
\end{equation}
\begin{lema} \label{lem-def-second-moment} Consider the system $\Phi$ with $u(t)=0$ for each $t$. 
The conditioned second moment $\mathfrak{X}(t)\in\mathcal{C}^{n}$ 
is given by $\mathfrak{X}(0)=\,(\Delta,\ldots,\Delta)$ and
\begin{equation} \label{eq_dynamical_x_second}
\mathfrak{X}(t+1)=\,\mathcal{U}_{A}(\mathfrak{X}(t)),\quad t=0,1,\ldots\ell-1.
\end{equation}
\end{lema}
\begin{proof} For a fixed, arbitrary $i\in\{0,\ldots,N\}$, note that
$$\mathfrak{X}_i(0)=\mathbb{E}\{x_0x'_0\,|\,\theta(0)=i\}=\mathbb{E}\{x_0x'_0\}=\Delta.$$
\comenta{For $t\geq 1$, we use  (\ref{eq-def-control}) in  $\Phi$ to get that $x$
evolves along the ``\emph{closed-loop system}'', as follows:
$x(t+1)=A_{\theta(t)}(t)x(t)$,
where $A_i(t)$ is described in (\ref{eq_At_bt}).}
From \eqref{eq-def-sys}, \eqref{eq-def-cond-second-moment} and the total probability law we obtain:
\begin{equation}\label{eq_ayuda_1}
\begin{aligned}
&\mathfrak{X}_i(t+1)=\mathbb{E}\{A_{\theta(t)}x(t)x(t)'A_{\theta(t)}'\,|\, \theta(t+1)=i\}\\
=&\sum_{j=1}^{N}\mathbb{E}\{A_{\theta(t)}x(t)x(t)'A_{\theta(t)}'\cdot \mathds{1}_{\theta(t)=j}\,|\, \theta(t+1)=i\}.
\end{aligned}
\end{equation}
In order to compute the right hand side of \eqref{eq_ayuda_1}, 
we need the following standard Markov chain property: for any
function $\Gamma:\theta(0),\ldots, \theta(t)\rightarrow \Re^{n,n}$ we have
$$\begin{aligned}
&\mathbb{E}\{\Gamma(\theta(0),\ldots, \theta(t)) \cdot \mathds{1}_{\theta(t)=j}\, |\,\theta(t+1)=i\}
\\&=\mathbb{E}\{\Gamma(\eta(\ell),\ldots, \eta(\ell-t)) \cdot \mathds{1}_{\eta(\ell-t)=j}\, |\,\eta(\ell-t-1)=i\}
\\&=\mathbb{E}\{\Gamma(\eta(\ell),\ldots, \eta(\ell-t)) | \eta(\ell-t)=j,\eta(\ell-t-1)=i\}
\\&\;\; \cdot \text{Prob}(\eta(\ell-t)=j\, |\,\eta(\ell-t-1)=i\}
\\&=\mathbb{E}\{\Gamma(\eta(\ell),\ldots, \eta(\ell-t)) | \eta(\ell-t)=j\} p_{ij}
\\& = \mathbb{E}\{\Gamma(\theta(0),\ldots, \theta(t))\,|\,\theta(t)=j \}  p_{ij},
\end{aligned}$$
then by replacing $\Gamma$ with $A_{\theta}(t)x(t)x(t)'A_{\theta(t)}'$
and applying the above in (\ref{eq_ayuda_1}) yields
\begin{equation}
\mathfrak{X}_i(t+1)=\sum_{j=1}^{N}p_{ij}A_j\mathfrak{X}_j(t)A_j'=\mathcal{U}_{A,i}(\mathfrak{X}(t)),\nonumber
\end{equation}
which completes the proof. \end{proof}
%
%
\begin{rema}\label{Rem_evolution_first_moment}
Let  $W(t)\in\mathcal{C}^{n}$ , $t=0,\ldots,\ell$ be given by
\begin{equation}\label{eq-def-W}
W_i(t)=\mathbb{E}\{x(t)x(t)'\cdot \mathds{1}_{\theta(t)=i}\}.
\end{equation}
This variable is commonly encountered in the majority of papers dealing with (standard) MJLS.
However, calculations similar to that in Lemma \ref{lem-def-second-moment} lead to
\begin{equation}\label{eq-simu-first-moment}
\begin{aligned}
&W_i(t+1)=\sum_{j=1}^{N}p_{ij}\frac{\pi_i(t+1)}{\pi_j(t)}A_jW_j(t)A_j'.
\end{aligned}
\end{equation}
Note that the Markov chain measure appears explicitly, leading to a time-varying
mapping from $W(t)$ to $W(t+1)$. The only exception is when the Markov chain is \emph{reversible}, in which case
the facts that $\pi_jp_{ji}=\pi_i p_{ij}$ and that the Markov chain starts with 
the invariant measure (by definition) yield
$$p_{ij}\frac{\pi_i(t+1)}{\pi_j(t)}=p_{ij}\frac{\pi_i}{\pi_j}=p_{ji},$$ 
in which case $W$ evolves exactly as in a standard MJLS. 
\comenta{for any $t$, hence eliminating the time 
dependence in \eqref{eq-simu-first-moment}.}
\comenta{XXX - acho que estah errado. Another remarkable fact is that, in a duality with the above, 
$\mathfrak X$ leads to time varying operators in the context of MJLS.}
\end{rema}
The following notion  
is adapted from \cite[Chapter 3]{CostaFragosoMarques05}. 
\begin{defi} We say that the system $\Phi$ with $u(t)=0$ is \emph{mean square stable} (MS-stable), whenever 
$$\lim_{\ell\to\infty} \mathbb{E}\left\{\|x(\ell)\|^2\right\}=0.$$
\end{defi}
This is equivalent to say that the variable 
$\mathfrak{X}(\ell)$ converges to zero as $\ell$ goes to infinity, leading 
to the following result. \comenta{in a straightforward manner.}
\begin{teo} The system $\Phi$ with $u(t)=0$ is MS-stable if and only if the spectral radius of $\mathcal{U}_A$ is smaller than one.
\end{teo}
%

\section{The TRM-JLQ problem}\label{sec-problem}
Let the output variable $y$ given by $y(\ell)=E_{\theta(\ell)}x(\ell)$ and 
\begin{equation}\nonumber
y(t)=
C_{\theta(t)}x(t)+ D_{\theta(t)}u(t),\quad t=0,\ldots,\ell-1.
\end{equation}
The TRM-JLQ consists of minimizing the mean square of $y$ with $\ell$ stages, 
as usual in jump linear quadratic problems, 
\comenta{A terminal 
cost is also included for generality} 
\begin{equation}\label{eq-functional-cost}
\min_{u(0),\ldots,u(\ell-1)} \mathbb{E}\left\{\sum_{t=0}^{\ell} \|y(t)\|^{2}\right\}.
\end{equation}
Regarding the information structure of the problem, we assume that $\theta(t)$
is available to the controller, that the control is in linear state feedback form, 
\begin{equation}\label{eq-def-control}
u(t)=K_{\theta(t)}(t)x(t), \quad t=0,1,\ldots\ell-1,
\end{equation}
where $K(t)\in\mathcal{C}^{m,n}$ is the decision variable, 
and that one should be able to compute the sequence $K(0),\ldots,K(\ell-1)$ 
prior to the system operation, that is, $K(t)$ is not a function of the 
observations $(x(s),\theta(s))$, $0\leq s\leq \ell$.
The conditioned second moment $\mathfrak{X}$ for the closed loop 
system is of much help in obtaining the solution.
The recursive formula for $\mathfrak{X}$ follows by a direct adaptation of 
Lemma \ref{lem-def-second-moment}, by replacing $A\in\mathcal{C}^{n}$ with 
its closed loop version 
\begin{equation}\label{eq_At_bt}
A_{i}(t)=A_i+B_iK_i(t).
\end{equation}
\begin{lema}\label{lem-def-second-moment-controlled-version}
The conditioned second moment $\mathfrak{X}(t)\in\mathcal{C}^{n}$ is given by 
$\mathfrak{X}(0)=(\Delta,\ldots,\Delta)$ and 
\begin{equation} \label{eq-controlled-second-moment}
\mathfrak{X}(t+1)=\mathcal{U}_{A(t)}(\mathfrak{X}(t)),\quad t=0,1,\ldots,\ell-1.
\end{equation}
\end{lema}
In what follows, for brevity we denote 
\begin{equation*}
\begin{aligned}
\mathcal Q(\ell)&=\pi(\ell)E'E,
\\ \mathcal{Q}(t)&=
C'C+K(t)'D'DK(t),\quad t=0,\ldots,\ell-1.
\end{aligned}
\end{equation*}
\begin{lema}\label{lem_cost_function} The TRM-JLQ problem can be formulated as
\begin{equation}\label{eq-functional-cost-modified}
\begin{aligned}
\min_{K(0),\ldots,K(\ell-1)}&\left\{\sum_{t=0}^{\ell} \langle\pi(t)\mathcal{Q}(t),\mathfrak{X}(t)\rangle\right\}. 
\end{aligned}
\end{equation}

\comenta{\begin{equation}\label{eq-functional-cost-modified}
\begin{aligned}
\min_{K(0),\ldots,K(\ell-1)}&\sum_{t=0}^{\ell-1} \langle\pi(t)(C'C+K(t)'D'DK(t)),\mathfrak{X}(t)\rangle \\
&\qquad + \langle \pi(\ell)E'E,\mathfrak{X}(\ell)\rangle.
\end{aligned}
\end{equation}}
\end{lema}
\begin{proof} 
The mean square  of the terminal cost, $y(\ell)$ is:  
\begin{equation}\label{eq-v-l}
\begin{aligned}
&\mathbb{E}\{\|y(\ell)\|^2\}=\sum_{i=1}^{N} \mathbb{E}\{x(\ell)'E'_{\theta(\ell)}E_{\theta(\ell)}x(\ell)\cdot \mathds{1}_{\theta(\ell)=i}\}\\
&=\sum_{i=1}^{N}\text{Tr}(\pi_i(\ell)E'_{i}E_{i}\mathfrak{X}_i(\ell)
= \langle \pi(\ell)E'E,\mathfrak{X}(\ell)\rangle.
\end{aligned}
\end{equation}
Now, by a calculation similar as above leads to 
\begin{equation}
\mathbb{E}\{\|y(t)\|^2\}=\langle \pi(t)(C'C+K(t)'D'DK(t)),\mathfrak{X}(t)\rangle.\label{eq-v-t} 
\end{equation}
Substituting \eqref{eq-v-l} and  \eqref{eq-v-t}  into \eqref{eq-functional-cost} we obtain \eqref{eq-functional-cost-modified}.
\end{proof}
Let us denote the gains attaining \eqref{eq-functional-cost-modified} 
by $K^{\text{op}}(t)$. From a dynamic programming standpoint, we introduce 
value functions $V^t:\mathcal{C}^n \rightarrow \Re$ by: $V^{\ell}=\langle \pi(\ell)\mathcal{Q}(\ell),\mathfrak{X}(\ell)\rangle$
\comenta{$V^{\ell}=\langle \pi(\ell)E'E,\mathfrak{X}(\ell)\rangle$} and for $t=\ell-1,\ell-2,\ldots,0$, 
\begin{equation}\label{eq-value-function}
\begin{aligned}
\displaystyle V^t(\mathfrak X) 
& = \min_{K(t),\ldots,K(\ell-1)}\left\{\sum_{\tau=t}^{\ell} \langle\pi(\tau)\mathcal{Q}(\tau),\mathfrak{X}(\tau)\rangle\right\}, 
\end{aligned}
\end{equation}
\comenta{\begin{equation}\label{eq-value-function}
\begin{aligned}
\displaystyle V^t(\mathfrak X) 
& = \min_{K(t),\ldots,K(\ell-1)}\sum_{\tau=t}^{\ell-1} \langle\pi(\tau)(C'C 
\\& \quad +K(\tau)'D'DK(\tau)),\mathfrak{X}(\tau)\rangle 
+ \langle \pi(\ell)E'E,\mathfrak{X}(\ell)\rangle,
\end{aligned}
\end{equation}}
where $\mathfrak X(t)=\mathfrak X$ and  
$\mathfrak X(\tau)$, $\tau=t+1,\ldots,\ell$, satisfies \eqref{eq-controlled-second-moment}.
\begin{teo} \label{teo_optimal_gains} 
Define $P(t)\in \mathcal{C}^n$ and $M(t)\in \mathcal{C}^{m,n}$, $t=0,\ldots,\ell-1$, as follows. 
Let $P(\ell)=\pi(\ell)E'E$ and for each  $t=\ell-1,\ldots,0$ and $i=1,\ldots,N$, compute:   
if $\pi_i(t)=0$, 
$$M_i(t)=0\quad\text{and}\quad P_i(t)=0,$$
else (if $\pi_i(t)>0$), 
\begin{align}
R_i(t)&=(B'_{i} \mathcal{D}_{i}(P(t+1)) B_{i} + \pi_i(t)D'_iD_i), \nonumber
\\ M_{i}(t)&=R_i(t)^{-1}B'_{i} \mathcal{D}_{i}(P(t+1)) A_{i},  \label{eq-gains}
\\ P_i(t)&=\pi_i(t)C'_iC_i + A'_i\mathcal{D}_i(P(t+1))A_i \nonumber
\\ &\; -A'_i\mathcal{D}_i(P(t+1)) B_i R_i(t)^{-1}B'_i\mathcal{D}_i(P(t+1))A_i. \label{eq-ric-L}
\end{align}
\comenta{
$$R(t)=(B'_{i} \mathcal{D}_{i}(P(t+1)) B_{i} + \pi_i(t)D'_iD_i)$$
and 
\begin{equation}\label{eq-gains}
\begin{aligned}
M_{i}(t)&=R(t)^{-1}B'_{i} \mathcal{D}_{i}(P(t+1)) A_{i},
\end{aligned}
\end{equation}
\begin{equation}\label{eq-ric-L}
\begin{aligned}
&P_i(t)=\pi_i(t)C'_iC_i + A'_i\mathcal{D}_i(P(t+1))A_i\\
&-A'_i\mathcal{D}_i(P(t+1))B_i R(t)^{-1}B'_i\mathcal{D}_i(P(t+1))A_i.
\end{aligned}
\end{equation}
}
Then,  
\begin{equation}\label{relation_P_X}
V^{t}(\mathfrak X)=\langle P(t),\mathfrak{X}\rangle 
\end{equation}
and $K^{\text{op}}(t)=M(t)$, $t=0,\ldots,\ell$. 
\end{teo}

\begin{proof} 
We apply the dynamic programming approach for the costs defined in \eqref{eq-value-function} and 
the system in \eqref{eq-controlled-second-moment}, 
whose state is the variable $\mathfrak{X}:=\mathfrak X(t)$. It can be checked that 
$\mathcal{U}$ is the adjoint operator of $\mathcal{V}$, and consequently  
\begin{equation}\label{eq-recursivity}
\begin{aligned}
&V^{t}(\mathfrak X)=\min_{K(t)}\,\,\langle\pi(t)\mathcal Q(t), \mathfrak{X}\rangle 
+ \langle P(t+1),\mathcal{U}_{A(t)}(\mathfrak{X})\rangle,
\\& =\min_{K(t)}\,\,\langle\pi(t)\mathcal Q(t)
+ A(t)'\mathcal{D}(P(t+1))A(t),\mathfrak{X}\rangle.
\end{aligned}
\end{equation}
Let us decompose $\mathfrak{X}$ as follows; let the set of states 
$\theta$ having zero probability of being visited at time $t$ be denoted by
$$\mathcal{N}_t=\{i:\pi_i(t)=0\}.$$
We write  $\mathfrak{X}=\mathfrak{X}^\text{N} + \mathfrak{X}^\text{P}$,
where $\mathfrak{X}^\text{N}$ is such that $\mathfrak{X}_i^\text{N}=0$ for any $i\notin\mathcal{N}_t$, 
and in a similar fashion $\mathfrak{X}_i^\text{P}=0$ for $i\in\mathcal{N}_t$. 
We now show that the term 
$$\langle \pi(t)\mathcal Q(t) 
+ A(t)'\mathcal{D}(P(t+1))A(t), \mathfrak{X}^\text{N} \rangle$$  
is zero irrespectively of $K(t)$.  
First, note that for $i\in\mathcal{N}_t$ we have $\pi_i(t)\mathcal Q_i(t)\mathfrak{X}_i^\text{N}=0$. 
Second, for $i\in\mathcal{N}_t$ one can check that $\pi_j(t+1)=0$ for all $j$ such that 
$p_{ji}>0$, so that $P_j(t+1)=0$. This yields $\mathcal{D}_i(P(t+1))=0$. 
Bringing these facts together and recalling that by construction 
$\mathfrak{X}_i^\text{N}=0$ for all $i\notin\mathcal{N}_t$, we evaluate 
\begin{equation}\label{eq-trivial-term}
\langle \pi(t)\mathcal Q(t) + A(t)'\mathcal{D}(P(t+1))A(t), \mathfrak{X}^\text{N} \rangle=0.
\end{equation}
By substituting \eqref{eq-trivial-term}
into \eqref{eq-recursivity} we write
\begin{equation}\label{eq-recursivity-2}
\begin{aligned}
V^{t}(\mathfrak X)&=\min_{K(t)}\,\,\langle \pi(t)\mathcal Q(t) 
+ A(t)'\mathcal{D}(P(t+1))A(t), \mathfrak{X}^\text{P} \rangle.
\\ &=\min_{K(t)} \sum_{\{i\notin\mathcal{N}_t\}}\text{Tr}\big((\pi_i(t)\mathcal Q_i(t) 
\\ &\qquad  + A_i(t)'\mathcal{D}_i(P(t+1))A_i(t))\mathfrak{X}_i^\text{P}\big).
\end{aligned}
\end{equation}
By expanding some terms and after some algebra to complete the squares, we have
\begin{equation}\label{eq-aux}
\begin{aligned}
V^{t}(\mathfrak{X})&=
\min_{K(t)}\sum_{\{i\notin\mathcal{N}_t\}}\text{Tr}\big((\pi_i(t)C'_iC_i+ A'_i\mathcal{D}_i(P(t+1))A_i\\
+&(K_i(t)-O_i(t))' (B'_i\mathcal{D}_i(P(t+1))B_i+\pi_i(t)D'_iD_i)\\
\cdot&(K_i(t)-O_i(t))- O_i(t)' (B'_i\mathcal{D}_i(P(t+1))B_i
\\ &\quad\quad\quad\quad\quad\quad+\pi_i(t)D'_iD_i) O_i(t))\mathfrak{X}_i^{P}\big), 
\end{aligned}
\end{equation}
where $O_i(t)=M_i(t)$ as given in \eqref{eq-gains}.
This makes clear that the minimal cost is achieved by setting 
$K^\text{op}_i(t)=M_i(t)$, $i\notin\mathcal{N}_t$. 
Now, by replacing $K_i(t)$ with $K^\text{op}_i(t)$ in \eqref{eq-aux}, 
\begin{equation*}
\begin{aligned}
V^{t}(\mathfrak{X})&=\sum_{\{i\notin\mathcal{N}_t\}} \text{Tr}\big( P_i(t)\mathfrak{X}_i^{P}\big)
\end{aligned}
\end{equation*}
with $P_i(t)$ as given in \eqref{eq-ric-L}, $i\notin\mathcal{N}_t$. Finally, by 
choosing $P_i(t)=0$, $i\in\mathcal{N}_t$, we write
\begin{equation*}
\begin{aligned}
V^{t}(\mathfrak{X})&=\sum_{\{i\notin\mathcal{N}_t\}} \text{Tr}\big( P_i(t)\mathfrak{X}_i^{P}\big) 
+ \sum_{\{i\in\mathcal{N}_t\}} \text{Tr}\big( P_i(t)\mathfrak{X}_i^{N}\big), 
\\& = \langle P(t),\mathfrak{X}^P\rangle 
+ \langle P(t),\mathfrak{X}^N\rangle  
=\langle P(t),\mathfrak{X}\rangle,
\end{aligned}
\end{equation*}
which completes the proof.
\end{proof}
\comenta{\begin{rema}
We have to set $P_i(t)=0$ when $\pi_{i}(t)=0$. To see this, 
consider $\mathfrak{X}=\mathfrak{X}^\text{N}$, that is, 
pick an $\mathfrak{X}$ such that $\mathfrak{X}^\text{P}=0$. 
In this case, from \eqref{eq-recursivity} and \eqref{eq-trivial-term}
we evaluate $V^{t}(\mathfrak X)=0$, so that for \eqref{relation_P_X} 
to be valid we need to set $P_i(t)$ as indicated. 
This is in contrast with the gains: \eqref{eq-recursivity-2} makes clear that 
the problem is irrespective of $K_i(t)$ for any $i:\pi_{i}(t)=0$.
\end{rema}}

\section{The duality between the TRM-JLQ and the LMMSE for standard MJLS}\label{sec-duality}
We consider the LMMSE for standard MJLS as presented in \cite{Costa94,Tuesta03}. The problem 
consists of finding the sequence of sets of gains $K^{\text{f}}(t)$, $t=0,\ldots,\ell$, that minimizes 
the covariance of the estimation error $\tilde z(t)=\hat z(t) - z(t)$ 
when the estimate is given by a Luenberger observer in the form 
$$\hat z(t+1)=A_{\eta(t+1)}\hat z(t)+K^{\text{f}}_{\eta(t+1)}(t)(y(t)-L_{\eta(t)}\hat z(t)),$$
where $y(t)$ is the output of the MJLS
\begin{equation}
\begin{cases}
z(t+1)&=\,F_{\eta(t+1)}z(t)+G_{\eta(t+1)}\omega(t)\\
\quad y(t)&=\,L_{\eta(t+1)}z(t)+H_{\eta(t+1)}\omega(t)\\
\quad z(0)&=\,z_0,
\end{cases}
\end{equation}
and $\omega(t)$ and $z_0$ are i.i.d. random variables
satisfying  $\mathbb{E}\{\omega(t)\}=0$, $\mathbb{E}\{\omega(t)\omega(t)'\}=I$ 
and $\mathbb{E}\{z_0z_0'\}=\Sigma$. Moreover,
it is assumed that $L_iH'_i=0$ and $H_iH'_i>0$. 
We write $\upsilon_i(t)=\text{Prob}(\eta(t)=i)$, $i=1,\ldots,N$, so that 
it is the time-reverse of $\pi$, $\upsilon(t)=\pi(\ell-t)$.  
Note that we are considering the same problem as in \cite{Costa94,Tuesta03}, though 
our notation is slightly different: 
here we assume that $(y(t),\eta(t+1))$ is available for the filter to 
obtain  $\hat z(t+1)$ and the system matrices are indexed by $\eta(t+1)$, 
while in the standard formulation $(y(t),\eta(t))$ are observed at 
time $t$ and the system matrices are indexed by $\eta(t)$. 
This ``time shifting'' in $\eta$ avoids a cluttering in the duality relation.
Along the same line, instead of writing the filter gains as a  function of the variable 
$Y_i(t)=\mathbb E\{\tilde z(t)\tilde z(t)'\cdot \mathds{1}_{\{\eta(t+1)=i\}}\}$, 
given by the coupled Riccati difference equation \cite[Equation 24]{Tuesta03}
\begin{equation*}
\begin{aligned}
Y_i(t+1)&=\sum_{j=1}^{N} p_{ji} \big\{F_jY_j(t)F'_j+\upsilon_j(t)G_jG'_j- F_jY_j(t)L'_j 
\\ & \quad \cdot \big(L_j Y_j(t) L_j' +\upsilon_j(t)H_jH_j'\big)^{-1} L_jY_j(t)F'_j \big\} 
\end{aligned}
\end{equation*}
whenever $\upsilon_i(t)>0$ and $Y_i(t+1)=0$ otherwise,
in this note we use the variable $S_i(t)=\mathbb E\{\tilde z(t)\tilde z(t)'\cdot \mathds{1}_{\{\eta(t)=i\}}\}$ defined in \cite{Costa95}, 
leading to $Y(t)=\mathcal{D}(S(t)).$ Replacing this in the above equation, after some 
algebraic manipulation one obtains \cite[Equation 8] {Costa95}:
\begin{equation}\label{CAREF2}
\begin{aligned}
S_i(t+1)&=\upsilon_i(t)G_iG'_i + F_i\mathcal{D}_i(S(t))F'_i - F_i\mathcal{D}_i(S(t))L'_i\\
&\cdot ( L_i  F_i(S(t)) L'_i+\upsilon_i(t)H_iH'_i)^{-1}
L_i\mathcal{D}_i(S(t))F'_i,
\end{aligned}
\end{equation}
whenever $\upsilon_i(t)>0$ and $S_i(t+1)=0$ otherwise, with initial condition
$S_i(0) =\mathbb E\{\tilde z(0)\tilde z(0)'\cdot \mathds{1}_{\{\eta(0)=i\}}\}=\upsilon_i(0)\Sigma$. 
The optimal gains are given for $t=0,\ldots,\ell$ by
\begin{equation}\label{eq-filter-gains}
\begin{aligned}
K^{\text{f}}_{i}(t)&= F_{i} \mathcal{D}_{i}(S(t)) L_{i}'
         \left( 
               L_{i} \mathcal{D}_{i}(S(t)) L_{i}' + \upsilon_i(t)H_{i}H_{i}' \right)^{-1}
\end{aligned}
\end{equation}
whenever $\upsilon_i(t)>0$ and $K^{\text{f}}_{i}(t)=0$ otherwise.
The duality relations between the filtering and control problems are now evident by direct 
comparison between \eqref{eq-ric-L} and \eqref{CAREF2}.
$F_i, L_i, G_i$ and $H_i$  are replaced with 
$A_i', B_i', C_i'$ and $D_i'$, respectively. Moreover, comparing the initial 
conditions of the \ver{coupled Riccati difference equations}, we see $\Sigma$ replaced with $E'E$. 
Also, we note that $P(0),P(1),\ldots,P(\ell)$ are 
equivalent to $S(\ell),S(\ell-1),\ldots, S(0)$, with a similar relation for the 
gains $K^{\text{f}}_i$ and $K^{\text{op}}_i$. The Markov chains driving the filtering and control systems are time-reversed one to each other. 

\begin{rema}\label{rema-extending-to-time-varying}
Time-varying parameters can be included both in standard MJLS and in $\Phi$ 
by augmenting the Markov state as to describe the pair $(\theta,t)$, $1\leq \theta\leq N$, 
$0\leq t\leq \ell$, and considering a suitable matrix $P$ of higher 
dimension $N\times (\ell+1)$. 
Although this reasoning leads to a matrix $P$ of high dimension, 
periodic and sparse, it is useful to make clear that our results are readily adaptable to plants 
whose matrices are in the form $A_{\theta(t)}(t)$. 
Either by this reasoning or by re-doing all computations given in this note 
for time-varying plants, 
we obtain the following generalization of \cite[Table 6.1]{Davis84}.
\begin{table}[h]
\centering
\begin{tabular}{ c  c }
\hline \\ \null \vspace{-.5 cm}
\\ \vspace{.1cm} FILTERING of MJLS & CONTROL of $\Phi$
\\\hline 
\\ \null \vspace{-.5 cm}
\\ \vspace{.08 cm}
$F_i(t)$ & $A'_i(t)$\\ \vspace{.1 cm}
$L_i(t)$ & $B'_i(t)$\\\vspace{.1 cm}
$G_i(t)$ & $C'_i(t)$\\\vspace{.1 cm}
$H_i(t)$ & $D'_i(t)$\\\vspace{.1 cm}
$K^{\text{f}}_i(t)$ & $K^{\text{op}\,\prime}_i(\ell-t)$\\\vspace{.1 cm}
$S_i(0)=\upsilon_i(0)\Sigma$&$P_i(\ell)=\pi_i(\ell)E'_iE_i$\\\vspace{.1 cm}
$S_i(t)$ & $P_i(\ell-t)$ \\\vspace{.1 cm}
$\eta_i(t)$ &$\theta_i(\ell-t)$\\
\hline\vspace{.03 cm}
\end{tabular}
\caption{Summary of the Filtering/Control Duality. $t=0,\ldots,\ell$.}
\label{Table}
\end{table}
\end{rema}
%
\section{Concluding remarks}\label{sec-conclusions}
We have presented an operator theory characterization of 
the conditional second moment $\mathfrak X$, 
an MS stability test and formulas for the optimal control of system $\Phi$. 
The results have exposed some interesting relations with standard MJLS.  
For system $\Phi$ it is fruitful to use the \emph{true} conditional second moment $\mathfrak X$ 
whereas for standard MJLS one has to resort to the variable $W$ given in \eqref{eq-def-W}
to obtain a recursive equation similar to the ones expressed in the Lemmas 3.1 and 4.1. 
Moreover, these classes of systems are equivalent if and only if the Markov chain 
is revertible, as indicated in Remark 1. The solution of the TRM-JLQ problem is given in Theorem 4.1 
in the form of a coupled Riccati equation that can be computed backwards 
prior to the system operation, as usual in linear quadratic problems for linear systems. 
The result beautifully extends the classic duality between filtering and control
into the relations expressed in Table 1. 
\comenta{It is perhaps surprising that such a direct duality would arise from a 
simple time reversion of the Markov chain, with no need for extra assumptions 
on the Markov chain, which is not necessarily ergodic nor time-revertible.}



\end{document}